%% file: arxiv.tex
\documentclass[11pt]{article}

\usepackage[utf8]{inputenc}
\usepackage[T1]{fontenc}

\usepackage{geometry}

\usepackage{amsthm, amsmath, amssymb}

\usepackage{float}
\usepackage{xcolor}
\usepackage{thmtools}
\usepackage{thm-restate}
\usepackage{enumerate}
\usepackage{hyperref}
\usepackage{graphicx}

\let\origR\R
\usepackage{complexity}
\let\R\origR

\newcommand{\Cref}[1]{\ref{#1}}
\newcommand{\citep}[1]{\cite{#1}}

\usepackage{tabularx}
\usepackage{environ}
\makeatletter
\newcommand{\problemtitle}[1]{\gdef\@problemtitle{#1}}%
\newcommand{\problemparameter}[1]{\gdef\@problemparameter{#1}}
\newcommand{\probleminput}[1]{\gdef\@probleminput{#1}}
\newcommand{\problemquestion}[1]{\gdef\@problemquestion{#1}}
\NewEnviron{problemgen}{
\problemtitle{}\problemparameter{}\probleminput{}\problemquestion{}
\BODY
\par\addvspace{.5\baselineskip}
\noindent
\begin{tabularx}{\textwidth}{@{\hspace{\parindent}} l X c}
    \multicolumn{2}{@{\hspace{\parindent}}l}{\@problemtitle} \\
    \textbf{Input:} & \@probleminput \\
    \textbf{Output:} & \@problemquestion
  \end{tabularx}
  \par\addvspace{.5\baselineskip}
}
\NewEnviron{problemdec}{
\problemtitle{}\problemparameter{}\probleminput{}\problemquestion{}
\BODY
\par\addvspace{.5\baselineskip}
\noindent
\begin{tabularx}{\textwidth}{@{\hspace{\parindent}} l X c}
    \multicolumn{2}{@{\hspace{\parindent}}l}{\@problemtitle} \\
    \textbf{Input:} & \@probleminput \\
    \textbf{Question:} & \@problemquestion
  \end{tabularx}
  \par\addvspace{.5\baselineskip}
}

\usepackage[textsize=small, color=red!30]{todonotes}

\hypersetup{
  pdftitle = {Enumerating minimal solution sets for metric graph problems},
  pdfauthor = {Benjamin Bergougnoux, Oscar Defrain, and Fionn {Mc~Inerney}},
  colorlinks = true,
  linkcolor = black!30!red,
  citecolor = black!30!green
}

\newtheorem{theorem}{Theorem}[section]

\newtheorem{lemma}[theorem]{Lemma}
\newtheorem*{lemma*}{Lemma}

\newtheorem{question}[theorem]{Question}

\newtheorem{claim}[theorem]{\normalfont Claim}

\newcommand{\minres}{\textsc{MinResolving}}
\newcommand{\minstres}{\textsc{MinStrongResolving}} 
\newcommand{\mingeo}{\textsc{MinGeodetic}}

\newcommand{\defeq}{:=}
\newcommand{\vc}{\mathtt{vc}}

\newcommand{\ETH}{\textsf{ETH}}

\newcommand{\transenum}{\textsc{Trans-Enum}} 
 
\newcommand{\algoa}{{\sf A}} 
\newcommand{\algob}{{\sf B}} 

\newcommand{\Z}{\mathcal{Z}}

\newcommand{\h}{\mathcal{H}}

\DeclareMathOperator{\dist}{{\sf dist}}

\newcommand{\intv}[2]{\left \{ #1, \dots, #2 \right \}}

\renewenvironment{abstract}
{\small\vspace{-1em}
\begin{center}
\bfseries\abstractname\vspace{-.5em}\vspace{0pt}
\end{center}
\list{}{
\setlength{\leftmargin}{0.4in}%
\setlength{\rightmargin}{\leftmargin}}%
\item\relax}
{\endlist}

\title{%
Enumerating minimal solution sets\\ for metric graph problems\thanks{An extended abstract of this paper has been accepted at WG 2024~\cite{WG2024}.}}
\author{
Benjamin Bergougnoux\thanks{Institute of Informatics, University of Warsaw, Poland}\and
Oscar Defrain\thanks{LIS, Aix-Marseille Université, France}\and 
Fionn {Mc~Inerney}\thanks{Algorithms and Complexity Group, Technische Universit\"{a}t Wien, Austria}}

\date{June 06, 2024} 

\begin{document}

\maketitle

\vspace{-.5cm}
\begin{abstract}
    \input{abstract}
\vskip5pt\noindent{}{\bf Keywords:} algorithmic enumeration, hypergraph dualization, metric dimension, geodetic sets, strong metric dimension, resolving sets, matroids.
\vskip5pt\noindent{}{\bf MSC codes:} 05C30.
\end{abstract}

\input{content}

\paragraph{Acknowledgements.} This work was supported by the ANR project DISTANCIA (ANR-17-CE40-0015) and the Austrian Science Fund (FWF, project Y1329). The second author is thankful to Simon Vilmin for extensive discussions on the links between minimal geodetic sets enumeration and the enumeration of the flats of a matroid.

\footnotesize

\bibliographystyle{alpha}
\bibliography{main}

\end{document}

%% file: abstract.tex
    
Problems from metric graph theory like \textsc{Metric Dimension}, \textsc{Geodetic Set}, and \textsc{Strong Metric Dimension} have recently had a strong impact in parameterized complexity by being the first known problems in \NP\ to admit double-exponential lower bounds in the treewidth, and even in the vertex cover number for the latter, assuming the Exponential Time Hypothesis. 
We initiate the study of enumerating minimal solution sets for these problems and show that they are also of great interest in enumeration.
Specifically, we show that enumerating minimal resolving sets in graphs and minimal geodetic sets in split graphs are equivalent to enumerating minimal transversals in hypergraphs (denoted \textsc{Trans-Enum}), whose solvability in total-polynomial time is one of the most important open problems in algorithmic enumeration.
This provides two new natural examples to a question that emerged in recent works: for which vertex (or edge) set graph property $\Pi$ is the enumeration of minimal (or maximal) subsets satisfying $\Pi$ equivalent to \textsc{Trans-Enum}? As very few properties are known to fit within this context---namely, those related to minimal domination---our results make significant progress in characterizing such properties, and provide new angles to approach \textsc{Trans-Enum}.
In contrast, we observe that minimal strong resolving sets can be enumerated with polynomial delay.
Additionally, we consider cases where our reductions do not apply, namely graphs with no long induced paths, and show both positive and negative results related to the enumeration and extension of partial solutions.



%% file: content.tex
\section{Introduction}

Metric graph theory is a central topic in mathematics and computer science that is the subject of many books and articles, with far-reaching applications such as in group theory~\cite{Gromov1987,Agol13}, matroid theory~\cite{BCK18}, learning theory~\cite{ChChMc22,ChChMoWa22,CKP22,CCMR23}, and computational biology~\cite{BD92}. Two very well-studied metric graph problems that arise in the context of network design and monitoring are the metric dimension~\cite{Slater75,HM76} and geodetic set~\cite{harary1993} problems. 
There is a rich literature concerning these problems and their variants, see, e.g.,~\cite{Slater87,KarpovskyCL98,sebo04,BGLM08,BMM+19}, with applications being found in areas like network verification~\cite{BEE+06}, chemistry~\cite{J93}, and genomics~\cite{TL19}. In particular, the strong metric dimension problem~\cite{sebo04} 
has begun to gather momentum. In this paper, we study the algorithmic enumeration of \emph{all} minimal solution sets for the metric dimension, geodetic set, and strong metric dimension problems.

In the {\sc Metric Dimension} problem, given an undirected and unweighted graph $G$ and a positive integer~$k$, the question is whether there exists a subset $S\subseteq V(G)$ of at most $k$ vertices such that, for any pair of vertices $u,v\in V(G)$, there exists a vertex $w\in S$ with $\dist(u,w)\neq \dist(v,w)$. A set of vertices $S\subseteq V(G)$ that satisfies the latter property is known as a \emph{resolving set} of $G$. 
This problem was shown to be \NP-complete in Garey and Johnson's book~\cite{DBLP:books/fm/GareyJ79}. In the last decade, its complexity was greatly refined, with it being proven to be \NP-hard in, e.g., 
planar graphs~\cite{DiazPSL17}, interval graphs 
of diameter~$2$~\cite{FoucaudMNPV17b}, (co-)bipartite graphs, 
and split graphs~\cite{ELW15}.
On the positive side, 
it is polynomial-time solvable in, e.g., trees~\cite{Slater75}, 
cographs~\cite{ELW15}, 
outerplanar graphs~\cite{DiazPSL17}, and graphs of bounded feedback edge number~\cite{E15}.
Further, it is in \FPT\ when parameterized by the max leaf number~\cite{E15}, 
the treelength plus maximum degree~\cite{BelmonteFGR17}, the distance to cluster (co-cluster, resp.)~\cite{GKIST23}, the treedepth, and the clique-width plus diameter~\cite{GHK22}.
In chordal graphs, it is also in \FPT\ when parameterized by the treewidth~\cite{BDP23}. In contrast, it is \W[2]-hard parameterized by the solution size~$k$~\cite{HartungN13}, \W[1]-hard parameterized by the pathwidth plus maximum degree~\cite{BP21} and the   pathwidth~\cite{GKIST23}, and para-\NP-hard parameterized by the pathwidth alone~\cite{LM22}. 

In the {\sc Geodetic Set} problem, given a graph $G$ and a positive integer $k$, the question is whether there exists a subset $S\subseteq V(G)$ of at most $k$ vertices such that every vertex in $G$ is on a shortest path between two vertices of $S$. A set of vertices $S\subseteq V(G)$ that satisfies the latter property is known as a \emph{geodetic set} of~$G$. 
It is \NP-complete in co-bipartite graphs~\cite{ekim2012}, interval graphs~\cite{floISAAC20}, 
and diameter-$2$ graphs~\cite{floCALDAM20}, but polynomial-time solvable in well-partitioned chordal graphs~\cite{wellpart}, 
outerplanar graphs~\cite{mezzini2018}, distance-hereditary graphs \cite{KN16}, 
and proper interval graphs~\cite{ekim2012}. 
Also, it is \W[1]-hard parameterized by the solution size plus feedback vertex number plus pathwidth, but 
in \FPT\ when parameterized by the treedepth, the clique-width plus diameter, and the feedback edge number~\cite{KK22}. 
It is also in \FPT\ when parameterized by the treewidth in chordal graphs~\cite{floISAAC20}.

In the {\sc Strong Metric Dimension} problem, given a graph $G$ and a positive integer $k$, the question is whether there exists a subset $S\subseteq V(G)$ of at most $k$ vertices such that, for any two distinct vertices $u,v, \in V(G)$, there exists a vertex $w\in S$ with either $u$ belonging to a shortest $w$--$v$ path or $v$ belonging to a shortest $w$--$u$ path. A set of vertices $S\subseteq V(G)$ that satisfies the latter property is known as a \emph{strong resolving set} of $G$. 
There is a polynomial-time reduction from an instance $(G,k)$ of {\sc Strong Metric Dimension} to an instance $(G',k)$ of {\sc Vertex Cover}, where $V(G)=V(G')$ and the edges of $G'$ connect pairs of vertices that are so-called ``mutually maximally distant'' in $G$~\cite{OellermannP07}. Due to its algorithmic implications, this relationship was further studied in~\cite{DM17,KPRY18}.

Recently, 
the three above problems were shown to be important well beyond the field of metric graph theory by being the first known problems in \NP\ to admit conditional double-exponential lower bounds in the treewidth, and even the vertex cover number~($\vc$) for {\sc Strong Metric Dimension}~\cite{FGKLMST23}. 
Further, the authors of~\cite{FGKLMST23} proved that, unless the Exponential Time Hypothesis (\ETH) fails, all three problems do not admit kernelization algorithms that output a kernel with $2^{o(\vc)}$ vertices.
Such kernelization lower bounds were priorly only known for two other problems~\cite{DBLP:journals/siamcomp/CyganPP16,DBLP:conf/iwpec/ChandranIK16}.
Lastly, they provided matching upper bounds in~\cite{FGKLMST23}, and
the lower bound technique they developed was used to obtain similar results for other \NP-complete problems~\cite{CCMR23,chakraborty2024tight}. 

Despite these problems being well-studied from an algorithmic complexity point of view, they have yet to be studied from the perspective of enumeration. We remedy this by initiating the study of enumerating minimal solution sets---the gold standard for enumeration---for the following problems.

\begin{problemgen}
  \problemtitle{\textsc Minimal Resolving Sets Enumeration (\minres{})}
  \probleminput{A graph $G$.}
  \problemquestion{The set of (inclusion-wise) minimal resolving sets of $G$.}
\end{problemgen}

\begin{problemgen}
  \problemtitle{\textsc Minimal Strong Resolving Sets Enumeration (\minstres{})}
  \probleminput{A graph $G$.}
  \problemquestion{The set of (inclusion-wise) minimal strong resolving sets of $G$.}
\end{problemgen}

\begin{problemgen}
  \problemtitle{\textsc Minimal Geodetic Sets Enumeration (\mingeo{})}
  \probleminput{A graph $G$.}
  \problemquestion{The set of (inclusion-wise) minimal geodetic sets of $G$.}
\end{problemgen}

In the same manner that these problems had a strong impact in parameterized complexity~\cite{FGKLMST23}, we show that they are also of great interest in enumeration. 
Specifically, they relate to classical problems such as the enumeration of maximal independent sets in graphs, which admits a polynomial-delay\footnote{The different notions from enumeration complexity are defined in Section~\ref{sec:prelims}.}  algorithm \cite{TsukiyamaIAS77}, and the enumeration of minimal transversals in hypergraphs, whose solvability in total-polynomial time is one of the most important open problems in algorithmic enumeration~\cite{eiter1995identifying,eiter2008computational}.

In the minimal transversals enumeration problem, denoted \transenum{} and also known as hypergraph dualization, given a hypergraph $\h$, the goal is to list all (inclusion-wise) minimal subsets of vertices that hit every edge of $\h$.
The best-known algorithm for \transenum{} runs in incremental
quasi-polynomial time by generating the $i^\text{th}$ minimal transversal of $\h$ in time $N^{o(\log N)}$, where $N=|\h|+i$~\cite{fredman1996complexity}.
Ever since, a lot of effort has been made to solve the problem in total-polynomial time in restricted cases.
Notably, there are polynomial-delay algorithms for $\beta$-acyclic hypergraphs~\cite{eiter1995identifying} and hypergraphs of bounded degeneracy~\cite{eiter2003new} or without small holes~\cite{kante2018holes}, and 
incremental-polynomial-time algorithms for bounded conformality hypergraphs \cite{khachiyan2007dualization} and geometric instances~\cite{elbassioni2019global}.

Due to the inherent difficult nature of the problem, and since no substantial progress has been made in the general case since~\cite{fredman1996complexity}, 
\transenum{} has gained the status of a ``landmark problem'' in terms of tractability, in between problems admitting total-polynomial-time algorithms, and those for which the existence of such algorithms is impossible unless $\P=\NP$.
This has motivated the study of particular cases of problems that are known to be 
at least as hard\footnote{%
An enumeration problem is \emph{at least as hard as} another enumeration problem if a total-polynomial-time algorithm for the first implies a total-polynomial-time algorithm for the second; the problems are (polynomially) \emph{equivalent} if the reverse direction also holds.} %
as \transenum{}; see, e.g.,~\cite{eiter1995identifying,kante2014split, conte2019listing}.
One of the most successful examples is the case of minimal dominating sets enumeration, with many particular cases shown to admit total-polynomial-time algorithms~\cite{kante2014split,kante2015chordal,golovach2018lmimwidth,bonamy2019kt}.
On the other hand, for problems that are notably harder than \transenum{} and for which the existence of total-sub-exponential-time algorithms is open, adapting the algorithm of \cite{fredman1996complexity} as in~\cite{elbassioni2009algorithms, elbassioni2022dualization}, or using it as a subroutine as in \cite{adaricheva2017discovery, defrain2021translating, defrain2020dualization} has also proved fruitful.

In light of these results, a line of research that emerged from \cite{kante2014split} (see also, e.g.,~\cite{conte2019maximal,strozecki2019survey}) consists of exploring the following question.
\begin{question}\label{qu:equivalence}
    For which vertex (or edge) set graph property $\Pi$ is the enumeration of minimal (or maximal) subsets satisfying $\Pi$ equivalent to \transenum{}?
\end{question}
We make progress on Question~\ref{qu:equivalence} by surprisingly proving that \transenum{} is equivalent to \minres{} (in general graphs) and \mingeo{} in split graphs. Notably, this adds two natural problems to the short list of those known to have this property.
Curiously, this contrasts with \minstres{} that we show can be solved with polynomial delay.

Interestingly, we additionally show that \mingeo{} is a particular case of enumerating the minimal flats of the graphic matroid associated to $K_n$ that are transversals of a given hypergraph. 
To the best of our knowledge, the latter problem is open, and thus, this encloses \mingeo{} by two generation problems whose complexity statuses are unsettled.
Hence, disproving the equivalence between \mingeo{} and \transenum{} by, e.g., showing that the first problem does not admit a total-polynomial-time algorithm unless $\P=\NP$,\footnote{As \transenum{} 
is in $\QP$ (quasi-polynomial time) and it is believed that $\NP\not\subseteq \QP$.}
would imply that the aforementioned variant of flats enumeration is intractable. 

Finally, we observe that the difficulty of the problems we study is tightly related to the maximum length of an induced path in the graph. This motivates the study of these problems on graphs that do not contain long induced paths with the aim of showing that it is not possible to get even more restricted graph classes for which \minres{} and \mingeo{} restricted to these graph classes are at least as hard as \transenum{}. While enumerating minimal geodetic and resolving sets is harder than \transenum{} on $P_5$ and $P_6$-free graphs, respectively, we show that they admit linear-delay algorithms in $P_4$-free graphs using a variant of Courcelle's theorem for enumeration and clique-width~\cite{Courcelle09}.

Concerning the difficulties and novelties of our techniques, our main results are based on several reductions which---unlike classical \NP-hardness reductions---preserve as much as possible the set of solutions of the original instance. For example, two main difficulties with our reduction from \transenum{} to \minres{} are to ensure that there are (1) at most a polynomial number of minimal resolving sets that do not correspond to minimal transversals, and (2) only a few minimal resolving sets corresponding to the same minimal transversal.

\section{Preliminaries}\label{sec:prelims}

For a positive integer $n$, let $[n]:=\{1,\ldots,n\}$.
In this paper, logarithms are binary.
We begin with definitions from enumeration complexity and refer the reader to~\cite{johnson1988generating,strozecki2019survey}. An algorithm runs in \emph{total-polynomial} time if it outputs every solution  exactly once and stops in a time which is polynomial in the size of the input plus the output.
Moreover, if the algorithm outputs the $i^\text{th}$ solution in a time which is polynomial in the size of the input plus $i$, then it runs in \emph{incremental-polynomial} time.
An enumeration algorithm runs with \emph{polynomial delay} if before the first output, between consecutive outputs, and after the last output it runs in a time which is polynomial in the size of the input.
Clearly, an algorithm running with polynomial delay runs in incremental-polynomial time, and an incremental-polynomial-time algorithm runs in total-polynomial time.
Moreover, reductions between enumeration problems (as described in Footnote~5) can sometimes preserve polynomial delay; this will be the case in this paper and we refer to~\cite{strozecki2019survey} for more details on these aspects.

For graphs and hypergraphs, see~\cite{diestel2012graph} and~\cite{berge1984hypergraphs} for definitions that are not recalled below.
A \emph{hypergraph} $\h$ is a set of vertices $V(\h)$ together with a family of edges $E(\h)\subseteq 2^{V(\h)}$.
It is a \emph{graph} when each of its edges has size precisely two, and \emph{Sperner} if no two distinct edges $E,F\in \h$ are such that $E\subseteq F$.
A \emph{transversal} of $\h$ is a subset $T\subseteq V(\h)$ such that $E\cap T\neq \emptyset$ for all $E\in E(\h)$.
It is \emph{minimal} if it is inclusion-wise minimal.
The set of minimal transversals of $\h$ is denoted by $Tr(\h)$, and the problem of listing $Tr(\h)$ (given $\h$) by \transenum{}.

Given a graph $G$ and two vertices $x,y$, $\dist(x,y)$ is the length of a shortest $x$-$y$ path in $G$.
Two vertices $u,v$ are \emph{false twins} if their open neighborhoods $N(u)$ and $N(v)$ are equal, and \emph{twins} if they are also adjacent.
For $k\in \mathbb{Z}^+$, we call $P_k$ a path on $k$ vertices.
We say that a vertex is \emph{complete} (\emph{anti-complete}, resp.) to a subset $S\subseteq V(G)$ if it is adjacent (non-adjacent, resp.) to each vertex in $S$; a set $A\subseteq V(G)$ is said to be complete to $S$ if every vertex $x\in A$ is.
For resolving sets, we say that a vertex $x$ \emph{distinguishes} a pair $a,b$ of vertices if $\dist(a,x)\neq \dist(b,x)$. 
For geodetic sets, we say that a pair $x,y$ of vertices \emph{covers} a vertex $v$ if $v$ lies on an $x$--$y$ shortest path.
We say that a (strong) resolving set (geodetic set, resp.) is minimal if it is inclusion-wise minimal.

Given a hypergraph $\h$ on vertex set $\{v_1,\dots,v_n\}$ and edge set $\{E_1,\dots,E_m\}$, the \emph{incidence bipartite graph} of $\h$ is the bipartite graph with bipartition $V=\{v_1,\dots,v_n\}$ and $H=\{e_1,\dots,e_m\}$, with an edge between $v_i\in V$ and $e_j\in H$ if $v_i$ belongs to~$E_j$.
The \emph{non-incidence bipartite graph} of $\h$ is the graph with the same vertices, but where there is an edge between $v_i\in V$ and $e_j\in H$ if $v_i$ does \emph{not} belong to~$E_j$. 
Finally, the \emph{(non-)incidence co-bipartite graph} of $\h$ is the (non-)incidence bipartite graph of $\h$ where $V$ and $H$ are completed into cliques.
For illustrations of these constructions, see the reductions in the next sections.

\section{Resolving sets}

In this section, we prove that \transenum{} and \minres{} are equivalent, and that our reductions preserve polynomial delay.
As for \minstres{}, we show that it is equivalent to the enumeration of maximal independent sets in graphs, and hence, that it admits a polynomial-delay algorithm.

We first deal with the reduction from \minres{}.
It is clear from the definition of distinguishing a pair of vertices that the resolving sets of a graph $G$ are exactly the transversals of the hypergraph $\h$ with the same vertex set and with an edge $E_{ab}\defeq \{v\in V(G) : \dist(a,v)\neq\dist(b,v)\}$ for every pair $a,b$ of distinct vertices in $G$.
Since $\h$ has $n$ vertices and $O(n^2)$ edges, and as it can be constructed in polynomial time in $n$, we derive the following.

\begin{theorem}\label{thm:minres-transenum}
    There is a polynomial-delay algorithm for \minres{} whenever there is one for \transenum{}.
\end{theorem}

Let us now deal with the reduction from \transenum{}.
Let $\h$ be a hypergraph on vertex set $\{v_1, \dots, v_n\}$ and edge set $\{E_1, \dots, E_m\}$.
For convenience, in our proof, we will furthermore assume that $n$ and $m$ are powers of $2$ greater than $2$, and that no edge of $\h$ contains the full set of vertices.
Note that these assumptions can be conducted without loss of generality, in particular since it can be assumed that $\h$ is Sperner and an edge containing the full set of vertices would imply it is the only edge of $\h$.
We describe the construction of a graph on $O(n+m)$ vertices and $O(n^2+m^2)$ edges whose set of minimal resolving sets can be partitioned into two families where the first one has size $O(nm^2)$, and where the second roughly consists of $O(nm)$ copies of the minimal transversals of $\h$.
See Figure~\ref{fig:MD} for an illustration of the construction.

We start from the non-incidence co-bipartite graph of $\h$ with bipartition $V\defeq \{v_1,\dots,v_n\}$ and $H\defeq \{e_1,\dots,e_m\}$, to which we add a clique $H'\defeq \{e'_1,\dots,e'_m\}$ that we make complete to $V$.
Then $V$, $H$, and $H'$ are cliques, $v_i$ is adjacent to every $e'\in H'$, and it is adjacent to $e_j$ if and only if $v_i\not\in E_j$.
We construct two additional sets $U\defeq \{u_1,u'_1,\dots,u_{(\log n) + 1}, u'_{(\log n) + 1}\}$ and $W\defeq \{w_1,w'_1,\dots,w_{(\log m) + 1},w'_{(\log m) + 1}\}$ on $2(\log n) + 2$ and $2(\log m) + 2$ vertices, respectively.
We complete $U$ into a clique minus each of the edges $u_iu'_i$, $i\in[(\log n) + 1]$, and add to $W$ each of the edges $w_jw'_j$, $j\in[(\log m) + 1]$.
For an integer $j\in\mathbb{N}$, we shall note $I(j)$ the set of indices (starting from $1$) of bits of value $1$ in the binary representation of $j$.
Then, we connect each $v_i$, $i\in[n]$, to the vertices $u_k$ and $u'_k$ for every $k\in I(i)$, and each of $e_j$ and $e'_j$, $j\in[m]$, to the vertices $w_k$ and $w'_k$ for every $k\in I(j)$.
Observe that, by the nature of the binary coding, no element of $V$ is complete or anti-complete to~$U$, and the same can be said for $H\cup H'$ and $W$.
Note that this binary representation gadget is derived from ideas used in, e.g.,~\cite{GKIST22,FGKLMST23}.
Finally, we connect every vertex of $U$ to every vertex of $H\cup H'$, and connect every vertex of $W$ to every vertex of $V$.
This concludes the construction of our graph $G$. We start with easy observations.
The following exploits that the sets $U$ and $W$ are constituted of pairs of twins.

\begin{figure}[bt]
    \centering
    \includegraphics[scale=0.8]{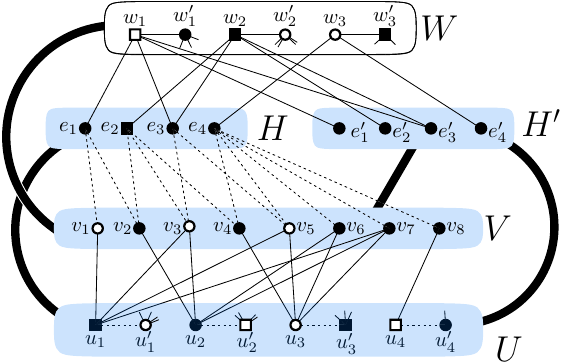}
    \caption{Illustration of the reduction from \transenum{} to \minres{} with $\h$ consisting of $E_1=\{v_1,v_2\}$, $E_2=\{v_2,v_3,v_4\}$, $E_3=\{v_3,v_5\}$, and $E_4=\{v_4,v_5,v_6,v_7,v_8\}$.
    Dashed lines represent non-edges, and a bold line between two sets of vertices $A,B$ means that $A$ is complete to $B$.
    For legibility, we do not show the edges of $G[U]$, which is almost a clique, nor the edges of the cliques $H$, $H'$, and $V$.
    We only show the non-edges between $V$ and $H$.
    We also do not fully represent some of the edges incident to the vertices $u'_i$ and $w_j'$.
    The set of white vertices is one of the $O(nm)$ minimal resolving sets associated to the minimal transversal $\{v_1,v_3,v_5\}$ of $\h$.
    The set of square vertices is one of the $O(nm^2)$ minimal resolving sets not associated with a minimal transversal.}
    \label{fig:MD}
\end{figure}

\begin{lemma}\label{lem:res-mandatory-twins}
    Let $S$ be a minimal resolving set of $G$.
    Then, $S$ intersects each of $\{u_i,u'_i\}\ (i\in[(\log n) + 1])$ and $\{w_j,w'_j\}\ (j\in[(\log m) + 1])$ on one vertex.
\end{lemma}

\begin{proof}
    First, note that each of the sets $\{u_i,u'_i\}$ $(i\in[(\log n) + 1])$ and $\{w_j,w'_j\}$ $(j\in[(\log m) + 1])$ defines a distinct pair of (false or not) twin vertices in $G$.
    As two (false or not) twins share the same distances to every other vertex in the graph, any resolving set must intersect them in order to distinguish them, and it intersects them on precisely one element whenever it is minimal.
\end{proof}

In the following, let us consider arbitrary
\begin{align*}
    X&\in \left\{\{x_1, \dots, x_{(\log n) + 1}\} : (x_1, \dots, x_{(\log n) + 1}) \in\,\,\!\prod_{i=1}^{(\log n) + 1} \hspace{.1cm}\{u_i, u'_i\}\hspace{.14cm}\right\},\\
    Y&\in \left\{\{y_1, \dots, y_{(\log m) + 1}\} : (y_1, \dots, y_{(\log m) + 1}) \in \prod_{j=1}^{(\log m) + 1} \{w_j, w'_j\}\right\},
\end{align*}
and denote by $\Z$ the set of all possible unions $Z\defeq X\cup Y$.
Note that there are $4nm$ possible choices for $Z$, and that by Lemma~\ref{lem:res-mandatory-twins}, any minimal resolving set contains one such $Z$ as a subset.
We characterize the pairs of vertices of $G$ that are distinguished by these sets.

\begin{lemma}\label{lem:Z-distinguishes}
    Let $Z\in \Z$ and let $P$ be the set of all pairs $\{e_j,e'_j\}$ with $j\in[m]$. 
    Then, $Z$ distinguishes a pair $a,b$ of distinct vertices if and only if $(a,b)\notin P$.
\end{lemma}

\begin{proof}
    We consider cases depending on the nature of each pair $a,b$ of distinct vertices in $G$ to show that they are distinguished by $Z$.
    Clearly if one of $a$ or $b$ belongs to $Z$, then the pair is distinguished.
    We thus assume $a$ and $b$ to be disjoint from $Z$ in the rest of the case analysis.
    
    We first consider $a\in U$.
    Then, $a\in \{u_i, u'_i\}$ for some $i\in[\log n + 1]$ and, 
    by assumption, $a\not\in Z$ and $a\neq x_i$.
    Then, $\dist(a, x_i)=2$, while $\dist(b, x_i)=1$ for any $b$ in $U$ or $H\cup H'$.
    If $b$ belongs to $V$ or $W$, then there is some $y_j\in Y$ such that $\dist(b, y_j)\leq 1$ and $\dist(a, y_j)\geq 2$.
    Hence, the pair $a,b$ is distinguished in this case.
    
    We now consider $a\in W$.
    Then, $a\in \{w_j, w'_j\}$ for some $j\in[\log m + 1]$ and, by assumption, $a\not\in Z$ and $a\neq y_j$. 
    If $b$ belongs to $W$, then $\dist(a,y_j)=1$ and $\dist(b,y_j)\geq 2$.
    If $b$ belongs to $H\cup H'$, then $\dist(a,x_i)\geq 2$ and $\dist(b,x_i)= 1$ for some $x_i\in X$.
    The same holds when $b$ belongs to $V$ as every element $v\in V$ is adjacent to some $x_i$ by the nature of the binary coding between $U$ and $V$.
    The case $b\in U$ was handled above.
    We conclude that the pair $a,b$ is distinguished in this case.
    
    Now assume that $a\in V$ and $b\in H\cup H'$. 
    Recall that, by the nature of the binary coding between $U$ and $V$, for each such $a$, there exists $x_i\in X$ such that $\dist(a,x_i)\geq 2$. 
    As $\dist(b,x_i)= 1$, the pair $a,b$ is distinguished by $x_i$ in this case.
    
    We are left with $\{a,b\}$ being a subset of $V$ or $H\cup H'$.
    In each of these cases, $a$ and $b$ have distinct adjacencies with respect to $U$ or $W$ as their indices within $V$ or $H\cup H'$ are distinct.
    In particular, this is true when $a\in H$ and $b\in H'$ or vice versa since $\{a,b\}\notin P$.
    Then, there exists $z\in Z$ such that $\dist(a,z)\neq \dist(b,z)$, and hence, the pair $a,b$ is distinguished, concluding the case.
    
    If $\{a,b\}\in P$, then no vertex in $Z$ distinguishes the pair $a,b$ as $U$ is complete to $H\cup H'$, each vertex in $W$ adjacent to $e_j$ is also adjacent to $e'_j$ for all $j\in[m]$, and each vertex in $W$ is at distance at most $2$ from each vertex in $H\cup H'$.
\end{proof}

Since by Lemma~\ref{lem:res-mandatory-twins} every minimal resolving set contains a choice of $Z$ as above, we get that the non-trivial part of minimal resolving sets in $G$ is dedicated to distinguishing pairs in $P$.
We characterize these non-trivial parts in the following.

\begin{lemma}\label{lem:res-garbage}
    If $S$ is a minimal resolving set of $G$ such that $S\cap (H\cup H')\neq\emptyset$, then $S=Z\cup \{e\}$ for some $Z\in \Z$ and $e\in H\cup H'$.
\end{lemma}

\begin{proof}
    Recall that, by Lemma~\ref{lem:res-mandatory-twins}, there exists $Z\in\Z$ such that $S\cap (U\cup W)=Z$.
    By Lemma~\ref{lem:Z-distinguishes}, only pairs $\{a,b\}\in P$  are not distinguished by $Z$.
    Since there is no edge between $H$ and $H'$, and as these sets are cliques, picking $e$ in any of these sets will satisfy $\dist(a,e)\neq\dist(b,e)$.
    Hence, $Z\cup \{e\}$ is a resolving set for every $e\in H\cup H'$, and the lemma follows by minimality.
\end{proof}

\begin{lemma}\label{lem:res-tr}
    If $S$ is a minimal resolving set of $G$ such that $S\cap (H\cup H')=\emptyset$, then $S=Z\cup T$ for some $Z\in \Z$ and some minimal transversal $T$ of $\h$.
\end{lemma}

\begin{proof}
    Let $\{e_j,e'_j\}\in P$.
    As by assumption $S\cap (H\cup H')=\emptyset$, then, by Lemma~\ref{lem:res-mandatory-twins}, $S=Z\cup T$ for some $T\subseteq V$.
    Now, for $v\in V$ to distinguish $e_j$ from $e_j'$, it must be that $v$ is non-adjacent to $e_j$ since $v$ is complete to $H'$.
    By construction, we deduce that $v\in E_j$ in that case.
    Since by Lemma~\ref{lem:Z-distinguishes} every pair in $P$ needs to be distinguished by $T$, we derive that $T$ is a transversal of $\h$.
    The minimality of $S$ implies that, for every $v\in V$, there exists at least one pair in $P$ that is distinguished by $v$ but not by $T\setminus \{v\}$.
    Hence, $T$ is a minimal transversal of $\h$.
\end{proof}

\begin{lemma}\label{lem:tr-res}
    If $T$ is a minimal transversal of $\h$, then $Z\cup T$ is a minimal resolving set of $G$ for any $Z\in \Z$.
\end{lemma}

\begin{proof}
    Since $T$ is a transversal of $\h$, every pair of $P$ is distinguished by a vertex of~$T$.
    By Lemmas~\ref{lem:res-mandatory-twins}~and~\ref{lem:Z-distinguishes}, we conclude that $Z\cup T$ is a resolving set.
    It is minimal by Lemma~\ref{lem:res-mandatory-twins} and the fact that, for every $v_i\in T$, there is some $E_j$ in $\h$ such that $(T\setminus\{v_i\})\cap E_j =\emptyset$, and hence, the pair $\{e_j,e_j'\}$ is not distinguished by $T\setminus \{v_i\}$, and by Lemma~\ref{lem:Z-distinguishes} this pair is not  distinguished by $(Z\cup T)\setminus \{v_i\}$.
\end{proof}

Note that, to every $T\in Tr(\h)$, there are $4nm$ corresponding distinct minimal resolving sets in $G$ obtained by extending $T$ with every possible $Z\in \Z$. 
We show that our reduction still preserves polynomial delay at the cost of (potentially exponential) space using a folklore trick on regularizing the outputs.

\begin{theorem}\label{thm:tr-minres}
    There is a polynomial-delay algorithm for \transenum{} whenever there is one for \minres{}.
\end{theorem}

\begin{proof}
    Let $\algoa{}$ be an algorithm for \minres{} running with polynomial delay $f(n)$ for some function $f:\mathbb{N}\to\mathbb{N}$, where $n$ is the number of vertices in $G$.
    We first describe an incremental-polynomial-time algorithm $\algob$ for \transenum{} generating the $i^\text{th}$ solution in $O(i\cdot (nm^2 \cdot f(n)))$ time.
    We start by constructing $G$ as above.
    Clearly, this can be done in polynomial time in $n+m$.
    Then, we simulate $\algoa{}$ on~$G$. 
    Each time $\algoa$ produces a set of the form $S=Z\cup T$ with $|T|\geq 2$, we check whether $T$ has already been obtained before by keeping every such $T$ in memory, and output it as a solution for \transenum{} if not.
    This concludes the description of $\algob{}$. 
    Its correctness follows from Lemmas~\ref{lem:res-garbage}, \ref{lem:res-tr}, and \ref{lem:tr-res}.

    Let us analyze the complexity of $\algob$.    
    By Lemma~\ref{lem:res-garbage}, $\algoa$ generates at most $O(nm^2)$ solutions in total time $O(nm^2\cdot f(n))$ before generating a first solution of the form $Z\cup T$ with $T\in Tr(\h)$.
    Hence, the first solution of $\algob$ is obtained in $O(nm^2\cdot f(n))$ time, as required.
    Suppose now that $\algob{}$ has produced $i$ solutions $T_1,\dots,T_i\in Tr(\h)$ in $O(i\cdot (nm^2 \cdot f(n)))$ time.
    By Lemmas~\ref{lem:res-garbage} and \ref{lem:res-tr}, when $\algob{}$ produces the $(i+1)^\text{th}$ solution for \transenum{}, the simulation of $\algoa$ has generated at most $i\cdot 4nm + O(nm^2)$ solutions of the form $Z\cup \{e\}$ with $e\in H\cup H'$ or $Z\cup T$ for \mbox{$T\in \{T_1,\dots, T_{i}\}$}.
    This takes $O\big((i\cdot 4nm + nm^2)\cdot f(n)\big)$ time by assumption, after which $\algob$ produces the next solution.
    Thus, in total, $\algob$ has spent $O(i\cdot 4nm \cdot f(n) + nm^2\cdot f(n))$ time outputting the $(i+1)^\text{th}$ solution of \transenum{} as desired.
    
    In the incremental time of $\algob$, the dependence on $i$ is linear.
    Using a folklore trick on regularizing the delay of such algorithms (see, e.g.,~\cite[Proposition~3]{capelli2023geometric}), we can regularize algorithm $\algob$ to polynomial-delay by keeping each new set $T$ in a queue, and pulling a new set from the queue every $O(nm^2\cdot f(n))$ steps.
\end{proof}

The space needed for the reduction of Theorem~\ref{thm:tr-minres} to hold is potentially exponential, as every obtained minimal transversal is stored in a queue.
However, using another folklore trick (see, e.g.,~\mbox{\cite[Section 3.3]{bonamy2020comp}}) on checking whether solutions have already been obtained by running the same algorithm on the same number of steps minus one, the reduction could preserve incremental-polynomial time and polynomial space at the cost of a worse dependence on the number of solutions. 
We end this section by dealing with \minstres{}.

\begin{theorem}
    \minstres{} can be solved with polynomial delay.
\end{theorem}

\begin{proof}
    It is known that, given a graph $G$, another graph $G'$ can be constructed in polynomial time such that the vertex covers of $G'$ are exactly the strong resolving sets of $G$~\cite[Theorem 2.1]{OellermannP07}.
    Moreover, the size of $G'$ is polynomial in the size of $G$, namely it satisfies $V(G')=V(G)$.
    As the vertex covers of a graph are the complements of its independent sets, we deduce a polynomial-delay algorithm for \minstres{} by the algorithm of Tsukiyama et al.~\cite{TsukiyamaIAS77}.
\end{proof}

\section{Geodetic sets}\label{sec:geo}

In this section, we prove that \transenum{} and \mingeo{} on split graphs are equivalent, and that our reductions preserve polynomial delay.
As for the general case, we show \mingeo{} to be a particular case of enumerating all the minimal flats of the graphic matroid associated to $K_n$ that are transversals of a given hypergraph, whose complexity status is unsettled to date.

We first deal with the reduction from \transenum{}.
Let $\h$ be a hypergraph on vertex set $\{v_1,\dots, v_n\}$ and edge set $\{E_1,\dots, E_m\}$.
We furthermore assume that $n,m\geq 1$ and that no vertex of $\h$ appears in every edge.
Note that these assumptions can be conducted without loss of generality.
In particular, if a vertex $v$ appears in every edge, then $Tr(\h)$ consists of $\{v\}$ and the minimal transversals of $\h'\defeq \{E\setminus \{v\}: E \in \h\}$, and so, solving \transenum{} on $\h$ is essentially equivalent to solving it on $\h'$, and thus, we can recursively remove such vertices.

We describe the construction of a split graph $G$ on $O(n+m)$ and $O(n^2m^2)$ edges whose set of minimal geodetic sets is partitioned into two families where the first has size $O(m)$ and the second is in bijection with the set of minimal transversals of $\h$. 
See Figure~\ref{fig:geodetic} for an illustration of the construction. 
We start from the non-incidence bipartite graph of $\h$ with bipartition $V\defeq \{v_1,\dots, v_n\}$ and $H\defeq \{e_1, \dots, e_m\}$, to which we add a set of vertices $U\defeq \{u_1, \dots, u_m\}$ with $u_j$ only adjacent to $e_j$ for each $j\in [m]$.
We then complete $U\cup V$ into a clique and add a vertex $e^*$ adjacent to every vertex in $V$. Finally, we add a vertex $u^*$ adjacent to every vertex in $G$.
This completes the construction.
We note that $G$ is a split graph with clique $K\defeq U\cup V\cup \{u^*\}$ and independent set $I\defeq H\cup \{e^*\}$.

Since $u^*$ is a universal vertex of $G$, the diameter of $G$ is at most 2, and we may reformulate $x$ being on a shortest $a$--$b$ path with $a\neq x\neq b$ as $x$ being the middle vertex of a $P_3$ in $G$ (as an induced subgraph).
We derive easy observations.

\begin{figure}[ht]
	\centering
	\includegraphics[scale=0.8]{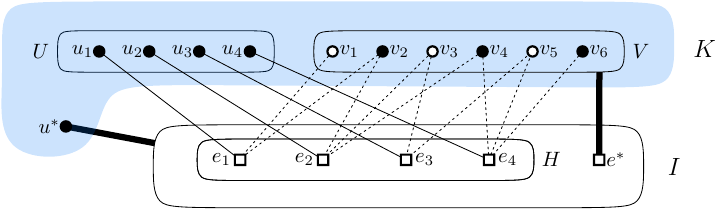}
	\caption{Illustration of the reduction from \transenum{} to \mingeo{} with $\h$ consisting of $E_1=\{v_1,v_2\}$, $E_2=\{v_2,v_3,v_4\}$, $E_3=\{v_3,v_5\}$, and $E_4=\{v_4,v_5,v_6\}$. 
    Dashed lines represent non-edges and the bold lines incident to $u^*$ and $e^*$ mean these two vertices are complete to $I$ and $V$, respectively.
    For legibility, we do not represent the edges of the clique $K$. The square vertices belong to any geodetic set. The set of white vertices is a minimal geodetic set obtained from the minimal transversal $\{v_1,v_3,v_5\}$ of $\h$.}
	\label{fig:geodetic}
\end{figure}

\begin{lemma}\label{lem:geo-H-included}
    The set $I$ is contained in every geodetic set of $G$.
\end{lemma}

\begin{proof}
    This holds since no vertex in $I$ is the middle vertex of a $P_3$ in $G$.
\end{proof}

\begin{lemma}\label{lem:geo-U-missing}
    Only the vertices in $U$ are not covered by the pairs of vertices in $I$.
\end{lemma}

\begin{proof}
    Clearly, all the vertices of $I$ are self-covered.
    Recall that $\h$ is assumed to contain at least one edge, and no vertex of $\h$ appears in every edge.
    Thus, if $x\in V\cup \{u^*\}$, then there exists $e\in H$ adjacent to $x$ and the pair $e,e^*$ covers it.
    Now, since no $P_3$ having its endpoints in $I$ contains a vertex of $U$, we conclude that only the vertices in $U$ are not covered by the pairs of vertices in $I$, as desired.
\end{proof}

\begin{lemma}\label{lem:u-dinstinguishes}
    The set $I\cup\{u\}$ is a minimal geodetic set of $G$ for every $u\in U$.
\end{lemma}

\begin{proof}
    This follows by Lemmas~\ref{lem:geo-H-included} and~\ref{lem:geo-U-missing} by observing that, for any $u'\in U$ with $u'\neq u$, we have a $P_3$ $eu'u$ for $e$ being the unique neighbor of $u'$ in $H$.
\end{proof}

We may now characterize minimal geodetic sets that are of interest as far as the transversality of $\h$ is concerned.

\begin{lemma}\label{lem:geo-tr}
    Let $S$ be a minimal geodetic set of $G$ such that $S\cap U=\emptyset$.
    Then, $S\cap K$ is a minimal transversal of $\h$.
\end{lemma}

\begin{proof}
    By Lemmas~\ref{lem:geo-H-included} and \ref{lem:geo-U-missing}, we have that $I \subseteq S$, and that only the elements in $U$ are not covered by the pairs of vertices in $I$.
    Let $T\defeq S\cap K$.
    Since $S\cap U=\emptyset$ and $u^*$ is adjacent to every vertex in $G$, we derive that $T\subseteq V$.
    Now, in order to cover $u_j$  $(j\in[m])$ it must be that a vertex from $T$ is not adjacent to $e_j$, as the only $P_3$ having $u_j$ as a middle vertex contains $e_j$.
    Hence, $T$ defines a transversal of $\h$ whenever the pairs of vertices in $T$ cover every such $u_j$.
    If $T$ was not minimal, then removing a vertex $v$ from $T$ would still intersect every edge of $\h$, which in turn would still cover every $u\in U$, a contradiction. 
\end{proof}

\begin{lemma}\label{lem:tr-geo}
    If $T$ is a minimal transversal of $\h$, then $T\cup I$ is a minimal geodetic set of~$G$.
\end{lemma}

\begin{proof}
    By Lemma~\ref{lem:geo-U-missing}, the pairs of vertices in $I$ cover every vertex of $G$ except for those in $U$.
    Now, since $T$ is a transversal, for every $E_j\in \h$, there exists $v\in T$ such that $v\in E_j$, and it follows that $v$ is not adjacent to $e_j$ and $e_ju_jv$ defines a $P_3$.
    Thus, every $u_j$ is covered and we conclude that $S\defeq T\cup I$ is a geodetic set.
    Let us assume that it is not minimal and let $x\in S$ such that $S\setminus \{x\}$ is still a geodetic set.
    By Lemma~\ref{lem:geo-H-included}, it cannot be that $x\in I$, and thus, it must be that $x\in T$.
    Then, for every $u_j$ $(j\in[m])$, there exists a pair $e_j,v$ with $v\in T\setminus \{x\}$ such that $e_ju_jv$ forms a $P_3$.
    Hence, for every $E_j\in \h$, there exists $v\in T\setminus \{x\}$ such that $v\in E_j$, a contradiction to the minimality of $T$.
\end{proof}

\begin{theorem}\label{thm:tr-mingeo}
	There is a polynomial-delay algorithm for \transenum{} whenever there is one for \mingeo{} on split graphs.
\end{theorem}

\begin{proof}
    This is a consequence of the fact that the graph $G$ can be constructed in polynomial time in the size of $\h$, has polynomial size, and that it contains $m$ minimal geodetic sets $I\cup\{u\}$ with $u\in U$. 
    All the other minimal geodetic sets of $G$ are of the form $I\cup T$ where $T$ is a minimal transversal of $\h$.
    Hence, a polynomial-delay algorithm for \mingeo{} would take at most $m$ times its delay between two consecutive minimal geodetic sets of the form $I\cup T$.
\end{proof}

We now argue that a polynomial-delay algorithm for \transenum{} yields one for \mingeo{} on split graphs.
Let $G$ be a split graph of bipartition $(K,I)$ with $K$ the clique and $I$ the independent set.
Among all such partitions we consider one that maximizes the size of $I$ and we may furthermore assume that $|I|\geq 2$, as otherwise the instance is trivial.
As in Lemma~\ref{lem:geo-H-included}, let us first note that $I\subseteq S$ for any geodetic set $S$ of $G$ as the neighborhood of every vertex $x\in I$ is a clique.
By the maximality of $I$, every vertex $v\in K$ that is not covered by a pair of vertices in $I$ has precisely one neighbor $u\in I$ and the set of vertices at distance at most $2$ from $u$ contains the full set $I$ as a subset.
Indeed, if it was not the case, then $v$ would be covered by the pair $u,w$, where $w\in I$ is a vertex at distance three from $u$.
Thus, to cover $v$ we must either pick $v$ or intersect $K\setminus N(u)$ the non-neighborhood of $u$ in $K$.
We identify all such vertices $v_1,\dots, v_k$ and their only neighbors $u_1,\dots,u_k$ in $I$ to construct a hypergraph $\h$ on vertex set $K$ with an edge $E_i=\{K \setminus N(u_i)\} \cup \{v_i\}$ for every $i\in [k]$.
We note that possibly $u_i=u_j$ for distinct $i,j\in[k]$, which is of no concern in the following.
Clearly, the construction can be achieved in polynomial time in the size of $G$, and by the above remarks we obtain a bijection between the minimal transversals of $\h$ and the minimal geodetic sets of $G$.
This yields the next theorem.

\begin{theorem}\label{thm:geo-split}
    There is a polynomial-delay algorithm for \mingeo{} on split graphs whenever there is one for \transenum{}.
\end{theorem}

We now end the section by showing that the general case of \mingeo{} reduces to enumerating all the minimal flats of the graphic matroid associated to the clique $K_n$ that are transversals of a given $n$-vertex hypergraph.

Let us start with the construction.
We consider a graph $G$ and construct a hypergraph $\h$ whose vertices are (unordered) pairs of distinct vertices of $G$, denoted $uv$ instead of $\{u,v\}$ for convenience, and where every vertex $v$ of $G$ gives rise to an edge $E_v\defeq \{xy : v\text{ is on a shortest $x$--$y$ path}\}$ in $\h$.
To avoid ambiguity, we shall refer to the vertices of $\h$ as \emph{nodes}, and use variables $r,s,t$ for nodes in the following.
Then, $\h$ has $O(n^2)$ nodes and $O(n)$ edges.
Clearly, every transversal $T$ of $\h$ induces a geodetic set $\bigcup_{t\in T} t$ of $G$, as every $E_v$ is hit by a pair $t$ in $T$, and that pair covers $v$ in $G$.
Unfortunately, minimal transversals of $\h$ do not necessarily define minimal geodetic sets of $G$ in that way, and not every minimal geodetic set of $G$ defines a minimal transversal of $\h$ by considering all the pairs of elements contained in it. 
Consider for example the graph $G$ obtained from a triangle $abc$ by adding a pendent vertex $d$ adjacent to $c$. 
Then, $\{a,b,d\}$ is the only minimal geodetic set of $G$, while $\{ab,ad,bd\}$ is easily verified to be a transversal of $\h$ that is \emph{not} minimal as it contains the transversal $\{ad, bd\}$ as a subset. On the other hand, $\{ac, bd\}$ can be checked to be a minimal transversal of $\h$, while $\{a,b,c,d\}$ is not a minimal geodetic set of $G$.
We nevertheless show that consistent sets that are transversals of $\h$ are in bijection with the geodetic sets of $G$ for an appropriate notion of consistency.

In the following, we call a subset $U$ of nodes of $\h$ \emph{consistent} if, whenever two distinct nodes $r,s \in U$ are such that $r\cap s\neq \emptyset$, then the unique other node $t$ such that $r\cup s = s\cup t = r\cup t$ is also part of $U$. 
In other words, a set of pairs is consistent if and only if it is the family of all pairs of some set.
As an example, a subset $U$ containing $ab$ and $bc$ but not $ac$ is not consistent, while the set $U=\{ab, ad, bd\}$ or the set of all nodes of $\h$ are consistent.
More generally, the family of all pairs of a given set is consistent.
The aforementioned correspondence is the following.

\begin{theorem}\label{thm:consistent-transversals}
    There is a bijection between the minimal geodetic sets of $G$ and the minimal consistent subsets of nodes that are  transversals of $\h$.
\end{theorem}

\begin{proof}
    Let $S$ be a minimal geodetic set of $G$ and consider the set $T$ of all pairs of vertices in $S$.
    Since every vertex $v$ in $G$ is covered by a pair of vertices in $S$, every edge $E_v$ in $\h$ is hit by a pair of $T$.
    As $T$ is consistent by construction, we conclude that it is a consistent transversal of $\h$.
    Let us assume toward a contradiction that it is not minimal with that property, and let $T'$ be a minimal consistent proper subset of $T$ that is a transversal of $\h$.
    Let $S'$ be the union of all pairs in $T'$.
    As $T'\subset T$ and both $T$ and $T'$ are consistent, $S'\subset S$.
    Then, by the minimality of $S$, there must be a vertex $v$ in $G$ that is not covered by any pair of $S'$.
    As $T'$ is only constituted of the pairs of elements in $S'$, we conclude that $E_v$ is not intersected by~$T'$, and hence, that it is not a transversal, a contradiction.

    Let $T$ be a minimal consistent transversal of $\h$ and consider the union $S$ of all pairs in $T$.
    Since every edge $E_v$ in $\h$ is hit by a pair $t$ in $T$, each vertex $v$ in $G$ is covered by a pair of vertices in $S$.
    Thus, $S$ is a geodetic set of $G$.
    Let us assume that it is not minimal and let $x$ be such that $S'\defeq S\setminus \{x\}$ is a geodetic set.
    Consider the family $T'$ of pairs of $S'$.
    Since $S'$ is a geodetic set, every edge $E_v$ in $\h$ is hit by a pair in $T'$.
    However, by the construction of $T'$ and as $T$ is consistent, we derive that $T'\subset T$, contradicting the minimality of $T$.
\end{proof}

We now discuss the implications of Theorem~\ref{thm:consistent-transversals}.
Observe that the consistency of a subset of vertices of $\h$ as defined above may also be expressed as satisfying a set of implications $\Sigma\defeq \{r,s \to t : r\cup s = s\cup t = r\cup t\}$ in the sense that any subset containing the premise of an implication in $\Sigma$ must contain its conclusion.
It is well known that the consistent sets in that context are the closed sets of a lattice~\cite{birkhoff1940lattice,wild2017joy}.
In the particular case of the rules defined above, the lattice is in fact known to be the lattice of flats (subsets of edges that are maximal with respect to the size of their spanning trees) of the graphic matroid associated to the clique on $n$ vertices, or equivalently, to be the lattice of partitions of a finite $n$-element set~\cite{birkhoff1940lattice}.
Consequently, listing minimal consistent transversals in our context may be reformulated as the enumeration of the minimal flats of the matroid associated to the clique $K_n$ that are transversals of $\h$.
To the best of our knowledge, no output-quasi-polynomial-time algorithm is known for that problem.
It should however be noted that in the more general setting where $\Sigma$ is allowed to contain any implications with premises of size at most two, the enumeration is intractable as it generalizes the dualization in lattices given by implicational bases of size at most two
\cite{defrain2020dualization}.

\section{Graphs with no long induced paths}\label{sec:no-lip}

In the previous sections, we showed that \minres{} and \mingeo{} are tough problems as they are at least as hard as \transenum{}, arguably one of the most challenging open problems in algorithmic enumeration to date.
Furthermore, these reductions hold for graphs with no long induced paths.
Namely, it can be easily checked that Theorem~\ref{thm:tr-minres} holds for $P_6$-free graphs, while Theorem~\ref{thm:tr-mingeo} holds for $P_5$-free graphs.
This motivates the study of these problems on instances that do not contain long induced paths.

We show \mingeo{} and \minres{} to be tractable on $P_4$-free graphs using a variant of Courcelle's theorem for enumeration and clique-width~\cite{Courcelle09}. 
We assume the reader to be familiar with MSO logic and clique-width, and refer the reader to~\cite{CourcelleE12} for an introduction.

\begin{theorem}\label{thm:P4free}
    Both \mingeo{} and \minres{} restricted to $P_4$-free graphs admit linear-delay algorithms with a preprocessing using time $O(n\log n)$.
\end{theorem}

\begin{proof}
    We argue that our theorem is a consequence of the meta-theorem from \cite[Corollary 2]{Courcelle09} stating that:
    \begin{itemize}
        \item given a monadic second-order formula $\phi(X_1,\dots,X_k)$, and
        \item a clique-expression of width $p$ expressing a graph $G$,
    \end{itemize}
    we can enumerate in linear delay all the tuples $(A_1,\dots,A_k)\in V(G)^k$ such that $G\models \phi(A_1,\dots,A_k)$ after a preprocessing using time $O(n\log n)$.
    
    Note that, for each $d\in \mathbb{N}$, there exists a first order formula $\phi_{d}(x,y)$ of size $O(d^2)$ testing whether $\dist(x,y)=d$ by testing whether there exists a path of length $d$ between $x$ and $y$ and none of length at most $d-1$.
    Hence, for every $\Delta\in \mathbb{N}$, the monadic second-order formula $\psi(X)=\psi'(X)\land (\forall X'\subset X, \, \neg \psi'(X'))$, where
    $$\psi'(X)\defeq \forall y,z\, (y\neq z)\implies \exists x\in X \, \bigvee_{i\in \{0,\dots,\Delta\}} \phi_i(x,y)\land \neg \phi_i(x,z)$$ 
    has size $O(\Delta^2)$, and, for any graph $G$ whose connected components have diameter at most $\Delta$ and for every $S\subseteq V(G)$, we have $G\models \phi(S) $ if and only if $S$ is a minimal resolving set of $G$.
    We obtain a similar monadic second-order formula for minimal geodetic set by replacing $\psi'(X)$ by the following:
    $$\forall y\, \exists x,z\in X\, \bigvee_{i\in\{0,\dots,\Delta\}} \phi_d(x,z)\land \phi_d(x,y,z),$$
    where $\phi_{d}(x,y,z)$ of size $O(d)$ tests whether $y$ is in a path of length $d$ between $x$ and $z$.
    Hence, the meta-theorem from \cite[Corollary 2]{Courcelle09} leads to the next claim.
    
    \begin{claim}\label{claim:cw}
        Given a clique-width expression of bounded width that defines a graph $G$ whose connected components have bounded diameter, we can solve \mingeo{} and \minres{} with linear delay after a preprocessing using time $O(n\log n)$.
    \end{claim}
    
    Now, we observe that $P_4$-free graphs---a.k.a.~cographs---have clique-width at most 2 \cite{CourcelleO00}, and that a clique-expression of width at most 2 can be computed in linear time \cite{CorneilPS85}.
    Moreover, every connected component of a $P_4$-free graph has diameter at most 2.
    Hence, this theorem is a direct consequence of Claim~\ref{claim:cw}.
\end{proof}

Interestingly, Theorem~\ref{thm:P4free} outlines a dichotomy for \mingeo{} in the sense that the problem is tractable for $P_k$-free graphs when $k\leq 4$, and that it is harder than \transenum{} otherwise.
This relates to similar behaviors and a line of research that emerged in~\cite{bonamy2019kt} on classifying forbidden induced subgraphs for which the enumeration of minimal dominating sets is tractable or harder than \transenum{}.

\section{Extension for minimal geodetic sets}\label{app:geo}

Usually when dealing with an enumeration problem $\Pi$ that asks to list subsets of a ground set $\{v_1,\dots,v_n\}$, a naive approach is to check whether the solutions of $\Pi$ can be constructed element by element, deciding at each step whether we include $v_i$ $(i=1,\dots, n)$ or not in the partial solution, in a way that each partial solution eventually leads to a solution.
This classical approach can be regarded as an efficient particular case of the backtrack search technique~\cite{read1975bounds}, and is usually referred to as \emph{flashlight search}~\cite{boros2004algorithms, khachiyan2007enumerating, capelli2023geometric} as it roughly amounts to looking ahead in the search tree to see whether there are solutions, in order to explore relevant branches only.
It has proved to be successful for a wide variety of very structured problems or restricted instances~\cite{boros2004algorithms, khachiyan2007enumerating,strozecki2019efficient,defrain2019neighborhood}.
Formally, a problem $\Pi$ is known to admit a polynomial-delay algorithm whenever the following problem, known as the \emph{extension problem} for $\Pi$, can be solved in polynomial time in the size of the input.

\begin{problemdec}
  \problemtitle{\textsc{Extension Problem for} $\Pi$ (\textsc{Ext}-$\Pi$)}
  \probleminput{Two disjoint subsets $A,B$ of the ground set.}
  \problemquestion{Is there a solution $S$ to $\Pi$ such that $A\subseteq S$ and $S\cap B = \emptyset$.}
\end{problemdec}

Most of the time, unfortunately, the extension problem is \NP-hard. 
The case of \textsc{Ext}-\transenum{} makes no exception to this rule~\cite{boros1998dual}, even for $\h$ being the family of closed neighborhoods of restricted graph classes~\cite{kante2015chordal,bonamy2019kt}.
In fact, it is even known to be \W[3]-complete parameterized by $|A|$~\cite{BLASIUS2022192,CASEL202248}.
In the following, we will show that the same applies to \mingeo{} for co-bipartite graphs.
This may suggest that generating minimal geodetic sets in that graph class is non-trivial.

Let $(\h,A,B)$ be an instance of \textsc{Ext}-\transenum{}, where $\h$ is a hypergraph on vertex set $\{v_1, \dots, v_n\}$ and edge set $\{E_1, \dots, E_m\}$, and $A$ and $B$ are two disjoint subsets of vertices.
We furthermore assume that $n,m\geq 1$ and that $V(\h)$ is not a minimal transversal of $\h$.
Note that these assumptions can be conducted without loss of generality, as if $V(\h)$ is a minimal transversal, then it is the only one and this can be checked in polynomial time.
We describe the construction of a graph $G$ on $O(n+m)$ vertices and $O(n^2+m^2)$ edges, and two sets $A',B'\subseteq V(G)$ such that there exists a minimal geodetic set $S$ with $A'\subseteq S$ and $S\cap B'=\emptyset$ if and only if there exists a minimal transversal $T$ of $\h$ such that $A\subseteq T$ and $T\cap B=\emptyset$.

We start from the incidence co-bipartite graph of $\h$ with bipartition $V\defeq \{v_1,\dots,v_n\}$ and $H\defeq \{e_1,\dots,e_m\}$, to which we add three vertices $a,b,c$ with $a$ complete to $H$, $b$ complete to $V$ and adjacent to $a$, and $c$ complete to $H\cup V$ and adjacent to $a$.
The obtained graph is co-bipartite with bipartition $(H\cup \{a,c\}, V\cup \{b\})$.
Then, we set $A'\defeq A\cup \{a,b,c\}$ and $B'\defeq B\cup H$.
Note that the diameter of $G$ is at most $2$.
Hence, and as in Section~\ref{sec:geo}, we may reformulate $x$ being on a shortest $s$--$t$ path in $G$ with $s\neq x\neq t$ as $x$ being the middle vertex of a $P_3$.

\begin{figure}
    \centering
    \includegraphics[scale=0.8]{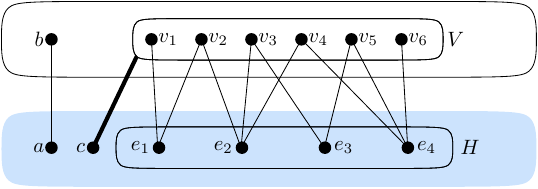}
    \caption{Illustration of the reduction from \textsc{Ext}-\transenum{} to \textsc{Ext}-\mingeo{} with $\h$ consisting of $E_1=\{v_1,v_2\}$, $E_2=\{v_2,v_3,v_4\}$, $E_3=\{v_3,v_5\}$, and $E_4=\{v_4,v_5,v_6\}$. 
    The bold line between $c$ and $V$ mean that $c$ is complete to $V$.
    For legibility, we do not represent the edges of the cliques $V\cup \{b\}$ and $H\cup \{a,c\}$.}
    \label{fig:geodeticExt}
\end{figure}

\begin{lemma}\label{lem:app-geo-tr}
    Let $S$ be a minimal geodetic set containing $A'$ and avoiding $B'$.
    Then, $S\cap V$ is a minimal transversal of $\h$ containing $A$ and avoiding $B$.
\end{lemma}

\begin{proof}
    We have that $S=\{a,b,c\}\cup T$ for some $T\subseteq V$ satisfying $A\subseteq T$ and $T\cap B=\emptyset$.
    Since the elements $e_j\in H$ may only be covered by pairs of the form $a,v$ for some $v\in V$ with $v\in E_j$, we conclude that $T$ is a transversal of~$\h$.
    Observe that $T$ is minimal because if $T\setminus \{v_j\}$ is a transversal of $\h$ for some $v_j\in T$, then every vertex in $G$ is still covered by the pairs of vertices in $S\setminus \{v_j\}$ (this is proved in the next lemma), which contradicts the minimality of $S$.
\end{proof}

\begin{lemma}\label{lem:app-tr-geo}
    Let $T$ be a minimal transversal of $\h$ containing $A$ and avoiding $B$.
    Then, $T\cup \{a,b,c\}$ is a minimal geodetic set of $G$ containing $A'$ and avoiding $B'$.
\end{lemma}

\begin{proof}
    Note that the pairs of vertices in $\{a,b,c\}$ cover every vertex in $V\cup \{a,b,c\}$.
    Now, since $T$ is a transversal, every $e_j\in H$ has a neighbor $v\in T$, and so, it is covered by $a,v$.
    Hence, $S\defeq T\cup \{a,b,c\}$ is a geodetic set.
    We argue that it is minimal.
    Note that $b$ and $c$ cannot be removed since by assumption $V$ is not a minimal transversal, and hence, there exists some $v\in V\setminus T$ that has to be covered.
    Indeed, no other $P_3$ than $bvc$ has its endpoints not in $H$ and $v$ as its middle vertex.
    Similarly, $a$ cannot be removed as otherwise the vertices in $H$ are not covered anymore.
    Suppose that $S\setminus \{v\}$ is a geodetic set for some $v\in V$.
    Then, every $e_j\in H$ is covered by a pair $a,w$ with $w\in T\setminus \{v\}$.
    We conclude that $T\setminus \{v\}$ is a transversal of $\h$, which contradicts the minimality of $T$.
\end{proof}

We conclude to the following theorem by Lemmas~\ref{lem:app-geo-tr} and \ref{lem:app-tr-geo}.

\begin{theorem}\label{thm:Ext-mingeo}
	The problem \textsc{Ext}-\mingeo{} is \NP-hard on co-bipartite graphs.
\end{theorem}

As a consequence, \mingeo{} should not admit a polynomial-delay algorithm using the classical flashlight search approach. 
We however note that this does not rule out the existence of a polynomial-delay algorithm for the problem as the extension problem is known to be hard for maximal cliques~\cite{brosse2020efficient}, yet the problem admits a polynomial-delay algorithm~\cite{TsukiyamaIAS77}.

\section{Extension for minimal resolving sets}\label{app:res}

In the following, we will show that the extension problem is \NP-hard for \minres{} in split graphs.
This may suggest that generating minimal resolving sets in that graph class is non-trivial.

Let $(\h,A,B)$ be an instance of \textsc{Ext}-\transenum{}, where $\h$ is a hypergraph on vertex set $\{v_1, \dots, v_n\}$ and edge set $\{E_1, \dots, E_m\}$ with $n,m\geq 1$, and $A$ and $B$ are two disjoint subsets of vertices.
Since adding a dummy vertex that is contained in exactly one dummy hyperedge of size $1$ simply ensures that any minimal transversal of $\h$ contains that vertex and only this dummy hyperedge is hit by this dummy vertex, we may furthermore assume that $\log(n+1)$ and $\log(m+1)$ are integers, and $E_m$ consists only of $v_n$ with $v_n\notin B$.
We describe the construction of a graph $G$ on $O(n+m)$ vertices and $O(n^2+m^2+nm)$ edges, and two sets $A',B'\subseteq V(G)$ such that there exists a minimal resolving set $S$ with $A'\subseteq S$ and $S\cap B'=\emptyset$ if and only if there exists a minimal transversal $T$ of $\h$ such that $A\subseteq T$ and $T\cap B=\emptyset$.

We start from the non-incidence bipartite graph of $\h$ with bipartition $V\defeq \{v_1,\dots,v_n\}$ and $H\defeq \{e_1,\dots,e_m\}$, to which we add a set of vertices $H'\defeq \{e'_1,\dots,e'_m\}$ that we make complete to $V$.
Moreover, we add four additional sets of vertices $U\defeq \{u_1,\dots,u_{\log(n + 1)}\}$, $U'\defeq \{u'_1,\dots,u'_{\log(n + 1)}\}$, $U^*\defeq \{u^*_1,\dots,u^*_n\}$, and $W\defeq \{w_1,\dots,w_{\log(m + 1)}\}$.
For an integer $j\in\mathbb{N}$, we shall note $I(j)$ the set of indices (starting from $1$) of bits of value $1$ in the binary representation of $j$.
We connect each $v_i$, $i\in\intv{1}{n}$, to the vertices $u'_k$ for every $k\in I(i)$, each of $e_j$ and $e'_j$, $j\in\intv{1}{m}$, to the vertices $w_k$ for every $k\in I(j)$, and each $u^*_i$, $i\in\intv{1}{n}$, to the vertices $u_k$ for every $k\in I(i)$.
Observe that, by the nature of the binary coding, no element of $V$ is anti-complete to~$U'$, and the same can be said for $H\cup H'$ and $W$, and $U^*$ and $U$.
For all $\ell\in\intv{1}{\log(n+1)}$, we also add an edge between $u_{\ell}$ and $u'_{\ell}$.
Finally, we add the necessary edges to make the vertices in $U'\cup U^*\cup H\cup H'$ into a clique.
This concludes the construction of our graph $G$.
The obtained graph is a split graph with the vertices of $U'\cup U^*\cup H\cup H'$ inducing a clique, and the vertices of $U\cup V\cup W$ inducing an independent set.
We set $A'\defeq A\cup U \cup W$ and $B'\defeq B\cup U' \cup U^* \cup H \cup H'$.

\begin{figure}
    \centering
    \includegraphics[scale=0.8]{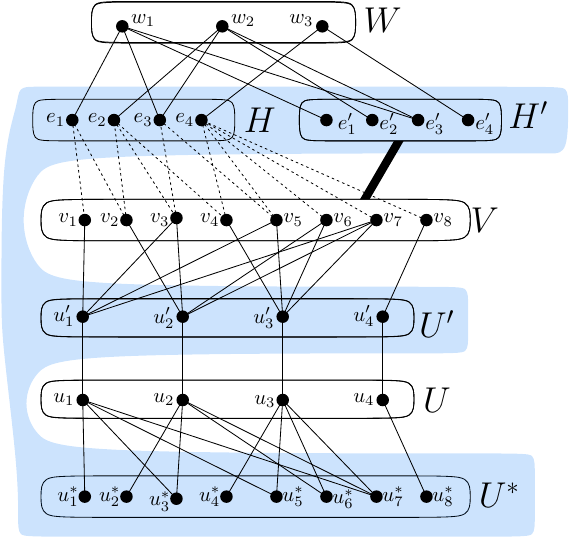}
    \caption{Illustration of the reduction from \textsc{Ext}-\transenum{} to \textsc{Ext}-\minres{} with $\h$ consisting of $E_1=\{v_1,v_2\}$, $E_2=\{v_2,v_3,v_4\}$, $E_3=\{v_3,v_5\}$, and $E_4=\{v_4,v_5,v_6,v_7,v_8\}$.
    The bold line represents the biclique between $H'$ and $V$.
    For legibility, we do not represent the edges of the cliques $H\cup H' \cup U' \cup U^*$.}
    \label{fig:MDExt}
\end{figure}

\begin{lemma}\label{lem:app-res}
Let $P$ be the set of all pairs $\{e_j,e'_j\}$ with $j\in\{1,\ldots,m-1\}$. Then, $Z=U\cup W\cup \{v_n\}$ distinguishes a pair $a,b$ of distinct vertices in $G$ if and only if $\{a,b\}\notin P$.
\end{lemma}

\begin{proof}
Clearly, if one of $a$ or $b$ belongs to $Z$, then the pair is distinguished. We thus assume $a$ and $b$ to be disjoint from $Z$ in the rest of the case analysis. If $a,b\in U'$, then there exists a vertex $u\in U$ such that $\dist(u,a)=1$ and $\dist(u,b)=2$, and hence, $u$ distinguishes the pair $a,b$. If $a,b\in U^*$ ($a,b\in V$, resp.), then since $a$ and $b$ do not have the same binary encoding, there exists some vertex in $U$ that is at distance $1$ ($2$, resp.) from one of $a$ and $b$ and distance $2$ ($3$, resp.) from the other, and hence, it distinguishes the pair $a,b$. Analogously, if $a,b\in H\cup H'$ and $\{a,b\}\notin P$, then there exists some vertex in $W$ that distinguishes the pair $a,b$.

If $a\in H\cup H$ and $b\in V\cup U' \cup U^*$, then there exists a vertex $w\in W$ such that $\dist(w,a)=1$ and $\dist(w,b)=2$, and hence, $w$ distinguishes the pair $a,b$. If $a\in U^*$ and $b\in U'$, then $\dist(v_n,a)=2$ and $\dist(v_n,b)=1$, and hence, $v_n$ distinguishes the pair $a,b$. Lastly, if $a\in U^*$ and $b\in V$ or $a\in U'$ and $b\in V$, then there exists a vertex $u\in U$ such that $\dist(u,a)=1$ and $\dist(u,b)\geq 2$, and hence, $u$ distinguishes the pair $a,b$.

Finally, if $\{a,b\} \in P$, then no vertex in $Z$ can distinguish the pair $a,b$ since every vertex in $U$ is at distance $2$ from both $a$ and $b$, $\dist(v_n,a)=\dist(v_n,b)=1$, each vertex in $W$ that is adjacent to $e_j$ is also adjacent to $e'_j$ for all $j\in \{1,...,m\}$, and each vertex in $W$ is at distance at most $2$ from each vertex in $H\cup H'$.
\end{proof}

\begin{lemma}\label{lem:app-res-tr}
    If $S$ is a minimal resolving set containing $A'$ and avoiding $B'$, then $S\cap V$ is a minimal transversal of $\h$ containing $A$ and avoiding $B$.
\end{lemma}

\begin{proof}
    Since $S$ avoids $B'$, we have that $S=U\cup W\cup T$ for some $T\subseteq V$ satisfying $A\subseteq T$ and $T\cap B=\emptyset$. Since by the construction of $G$ and Lemma~\ref{lem:app-res}, among the vertices in $V(G)\setminus B'$ that may be in $S$, the pair $e_j, e'_j$ can only be distinguished by a vertex $v\in S\cap V$ with $v\in E_j$, we conclude that $T$ is a transversal of~$\h$. It is minimal as if every $E_j\in\h$ is still intersected by $T\setminus \{v_j\}$ for some $v_j\in T$, then the pair $e_j, e'_j$ is still distinguished by a vertex in $S\setminus \{v_j\}$, which contradicts the minimality of $S$ by Lemma~\ref{lem:app-res} and the fact that $v_j\neq v_n$ since $v_n$ is necessarily in any minimal transversal of $\h$.
\end{proof}

\begin{lemma}\label{lem:app-tr-res}
    If $T$ is a minimal transversal of $\h$ containing $A$ and avoiding $B$, then $T\cup U \cup W$ is a minimal resolving set of $G$ containing $A'$ and avoiding $B'$.
\end{lemma}

\begin{proof}
    By construction, $v_n\in T$. For $P$ the set of all pairs $\{e_j,e'_j\}$ with $j\in\{1,\ldots,m-1\}$, by Lemma~\ref{lem:app-res}, $Z=U\cup W\cup \{v_n\}$ distinguishes every pair $a,b$ of distinct vertices in $G$ with $\{a,b\}\notin P$.
    Now, since $T$ is a transversal, for each $e_j\in H$, there exists a vertex $v\in V$ such that $\dist(v,e_j)=2$ and $\dist(v,e'_j)=1$, and so, $v$ distinguishes the pair $e_j,e'_j$. Hence, $S\defeq T\cup U\cup W$ is a resolving set. We argue that it is minimal. Note that, for any vertex $u\in U$ ($w\in W$, resp.), there exist two vertices in $U^*$ ($H'$, resp.) that are distinguished only by $u$ ($w$, resp.) among all the vertices in $S$ (in particular, they only differ in exactly one bit in their binary representation), and hence, no vertices from $U$ ($W$, resp.) can be removed from $S$. Suppose that $S\setminus \{v\}$ is a resolving set for some $v\in V$. Then, every pair $e_j,e'_j$ is distinguished by a vertex in $T\setminus \{v\}$ by Lemma~\ref{lem:app-res}. We conclude that $T\setminus \{v\}$ is a transversal, which contradicts the minimality of $T$.
\end{proof}

We conclude to the following theorem by Lemmas~\ref{lem:app-res-tr} and \ref{lem:app-tr-res}.

\begin{theorem}\label{thm:Ext-minres}
	The problem \textsc{Ext}-\minres{} is \NP-hard on split graphs.
\end{theorem}

As a consequence, \minres{} should not admit a polynomial-delay algorithm using the classical flashlight search approach.

\section{Perspectives for further research}

We investigated a number of problems related to the metric dimension that connect to problems of huge interest in algorithmic enumeration.
Except for \minstres{} that can be solved with polynomial delay on general graphs, we showed that \minres{} is equivalent to \transenum{} and that the same holds for \mingeo{} when restricted to split graphs.
Moreover, the general case of \mingeo{} may be seen as an intriguing variant of enumerating the flats of a matroid for which the complexity status is unsettled.

The results presented in this work showed that the difficulty of \minres{} and \mingeo{} is tightly related to the maximum length of an induced path in the graph at hand.
This motivates the study of these problems on $P_k$-free graphs for small values of $k$.
Except for \mingeo{} that we completely characterized with respect to \transenum{}, the case of \minres{} for $k=5$ is yet to be classified.
While it would be interesting to extend our hardness result (with respect to \transenum{}) for \minres{} to $P_5$-free graphs or even split or co-bipartite graphs, it is challenging as our construction in the proof of Theorem~\ref{thm:tr-minres} seems impossible to extend to these cases. 
We note that the case of co-bipartite graphs is also open for \mingeo{}.
For these graph classes, we were not able to devise total-polynomial-time algorithms; we however note that the extension problems for \minres{} and \mingeo{} are hard on split and co-bipartite graphs, respectively, which suggests that the generation is non-trivial in those cases.

Other open directions are to know whether \mingeo{} admits a total-quasi-poly\-nomial-time algorithm, or how it relates to problems known to be harder than \transenum{} that admit sub-exponential-time algorithms. Problems of interest include the dualization of products of posets~\cite{elbassioni2009algorithms} or the dualization in distributive lattices~\cite{elbassioni2022dualization}.

As for candidates for Question~\ref{qu:equivalence}, it is open for minimal connected dominating sets \cite{kante2014split,conte2019maximal}. 
This case was however conjectured not to be equivalent to \transenum{} by Kanté at the 2015 Lorentz Workshop on enumeration algorithms~\cite{2015Open}.